\documentclass[12pt]{article}

\usepackage[utf8]{inputenc}
\usepackage{jfe_preambule}

\begin{document}

\setlist{noitemsep}  

\title{Incremental Sharpe and other performance ratios}

\author{Eric Benhamou 
\thanks{LAMSADE, Université Paris Dauphine, Place du Maréchal de Lattre de Tassigny,75016 Paris, France} 
\textsuperscript{,} 
\thanks{A.I. SQUARE CONNECT, 35 Boulevard d'Inkermann 92200 Neuilly sur Seine, France}  
\textsuperscript{,} 
\thanks{E-mail: eric.benhamou@dauphine.eu, eric.benhamou@aisquareconnect.com} \hspace{0.2cm}, 
Beatrice Guez 
\textsuperscript{\dag}  
\textsuperscript{,}  
\thanks{E-mail: beatrice.guez@aisquareconnect.com}}

\date{}              


\singlespacing

\maketitle

\vspace{-.2in}
\begin{abstract}
\noindent We present a new methodology of computing incremental contribution for performance ratios for portfolio like Sharpe, Treynor, Calmar or Sterling ratios. Using Euler's homogeneous function theorem, we are able to decompose these performance ratios as a linear combination of individual modified performance ratios. This allows understanding the drivers of these performance ratios as well as deriving a condition for a new asset to provide incremental performance for the portfolio. We provide various numerical examples of this performance ratio decomposition.
\end{abstract}

\medskip



\thispagestyle{empty}

\clearpage

\onehalfspacing
\setcounter{footnote}{0}
\renewcommand{\thefootnote}{\arabic{footnote}}
\setcounter{page}{1}

\section{Introduction}

When facing choices to invest in various funds (whether mutual or hedge funds), it is quite common to compare their Sharpe ratio, or other performance ratios like Treynor or recovery ratio in order to rank funds. These ratios aim at measuring performance for a given risk. They achieve two important things: they measure performance taking into account risk. They allow constructing the optimal performance as the result of an optimization program.

The usual performance metric is the eponymous Sharpe ratio established in \cite{Sharpe_1966}. It is a simple number easy to derive and intuitive to understand as it computes the ratio of the excess return over the strategy standard deviation.  It has various limitations that have been widely emphasized by various authors (\cite{PilotteSterbenz_2006},  \cite{Sharpe_1998}, \cite{Nielsen_Vassalou_2004}) leading to other performance ratios like Treynor ratio (see \cite{Treynor_FBlack_1973}), but also Calmar (see \cite{Young_1991}), Sterling (see \cite{McCafferty_2003}) or Burke ratio (see \cite{Burke_1994}). Other authors have also tried to provide additional constraints to the Sharpe as in \cite{Bertrand_2008} or more recently in \cite{Darolles_2012} or to use option implied volatility and skewness as in \cite{DeMiguel_Plyakha_Uppal_Vilkov_2013}. There have been also numerous empirical work on Sharpe ratio as for the most recent ones  in \cite{Giannotti_Mattarocci_2013}, \cite{AndersonBianchi_2014}. There has been also an interesting approach by \cite{Challet_2017} to compute Sharpe ratio through total drawdown duration.

An important feature that has been noted in \cite{Darolles_2012} or \cite{Steiner_2011} but only for the Sharpe ratio is the fact that most of the performance ratios are so called 0 Euler homogeneous with respect to the portfolio weights. In financial terms, there are not sensitive to the leverage of the portfolio. 

The contribution of our paper is to exploit this mathematical property and re-derive well known results on the Sharpe ratio in a new manner. As a consequence, we obtain the condition for a new asset to increase the overall Sharpe of a portfolio. We also extend the incremental performance marginal sensitivity to all performance ratios that are 0 Euler homogeneous with respect to the portfolio weights. This allows in particular to understand the performance ratios drivers. We finally show how to decompose performance ratios between a benchmark and the individual portfolio constituents.

\section{Euler homogeneous functions and its application to performance ratios} \label{sec:EulerHomogeneousFunctions}

\subsection{Euler's theorem}

In mathematics, one call a homogeneous function one that has a multiplicative scaling behaviour. If we multiply all its arguments by a constant factor, then its value is multiplied by some power of this factor. If we denote by $f \colon \mathbb{R}^n \rightarrow \mathbb{R}$  a multidimensional function from $\mathbb{R}^n$ to $\mathbb{R}$, then the function $f$ is said to be homogeneous of degree $k$ if the following holds:

\begin{equation}
f( \alpha v ) = \alpha^k f(v)
\end{equation}

for all positive $\alpha > 0$ and $v \in \mathbb{R}^n$. If the function is continuously differentiable ( and this generalized also to almost surely continuously differentiable function), the Euler's homogeneous function theorem \footnote{This theorem is trivially proved by differentiating $f( \alpha v ) = \alpha^k f(v)$ with respect to $\alpha$ for the implication condition and by integrating the differential equation $x \cdot \nabla f(x)= k f(x)$ for the reverse condition} states that the function is homogeneous if and only if 

\begin{equation} \label{Euler}
x \cdot \nabla f(x)= k f(x)
\end{equation}

where $ \nabla f(x)$ stands for the gradient of $f$. This theorem (shown in various book like for instance \cite{Lewis_1969}) gives in particular a nice decomposition of any homogeneous function provided we can compute the gradient function as it says that the function is a linear combination of partial derivatives as follows:
\[
f(x) = \frac 1 k \sum_{i=1..n} x_i \frac{ \partial f }{ \partial x_i}
\]

\subsection{Intuition with Sharpe ratio}
Let us see how this can be applied to any homogeneous performance ratio. In order to build our intuition, we will start by the Sharpe ratio as this is a simple and well know ratio. We assume we have a portfolio of $n$ assets with weights $w_i$. We denote by $R_f$ the risk free rate and $R_p$ the portfolio return. The Sharpe ratio is defined as the fraction of the portfolio excess return $r_p$ over the portfolio volatility $\sigma_p$ and given by
\begin{equation}
S_p = \frac{R_p - R_f}{\sigma_p} = \frac{r_p}{\sigma_p} 
\end{equation}

If we decompose the portfolio excess return $r_p$ as the convex combination of its assets excess returns with percentage weights $ w_i$, we get that the Sharpe ratio is a convex combination of the modified Sharpe
\begin{equation}
S_p =  \frac{ \sum_{i=1}^{n} w_i r_i}{\sigma_p} = \sum_{i=1}^{n}  w_i   \frac{r_i}{\sigma_p}
\end{equation}

This says that if we were looking at a modified Sharpe ratio of each portfolio constituent where the volatility of the constituent is modified into the one of the portfolio, then the Sharpe ratio of the portfolio is simply the convex combination of these modified Sharpe ratio. This is nice from a theoretical point of view but not very useful as this forces us to compute the volatility of the portfolio and does not give any hindsight about asset $i$ volatility. This is where Euler homogeneous formula comes at the rescue. The Sharpe ratio like many other performance ratio has the particularity that it is the fraction of two homogeneous function of degree 1. The decomposition for the excess return in terms of a linear combination of the portfolio weight is obvious. 

More subtle is the fact that the volatility of the portfolio can also be decomposed as the convex sum of individual volatility contributions. Indeed, the volatility is an homogeneous function of degree 1 of the portfolio weights as scaling the weights by a factor increases the portfolio by the same factor. In the sequel, we will denote for asset $i$, $\rho_{i,p}$ its correlation with portfolio $P$, $\sigma_i$ its volatility and $S_i$ its Sharpe ratio. Thanks to Euler's homogeneous function theorem, we know that the portfolio volatility can be written as follows

\begin{proposition} \label{prop:port_vol}
The weighted marginal contributions to volatility sum up to portfolio volatility as follows:
\begin{eqnarray}
\sigma_p    =   \sum_{i=1}^n w_i \frac{ \partial \sigma_p} { \partial w_i }  = \sum_{i=1}^n w_i \rho_{i,p}  \sigma_i
\end{eqnarray}
\end{proposition}

\begin{proof} trivial consequence of Euler's homogeneous function theorem and given in \ref{proof:0}.
\end{proof} 

This results was first derived in \cite{Bertrand_2009} and in \cite{Maillard_2010} in the case of equally weighted portfolio and was also noted in \cite{Darolles_2012}. Euler property has been mentioned as early as in \cite{Litterman_1996} and \cite{Tasche_1999}. Using this first property, it is now easy to derive a convex combination for the portfolio Sharpe ratio as follows.

\begin{proposition} \label{prop:1}
The portfolio Sharpe ratio is a convex combination of individual Sharpe ratios weighted by the inverse of the asset $i$ correlation with portfolio $P$, $\rho_{i,p}$:

\begin{equation} \label{eq:Sharpe}
S_p = \sum_{i=1}^{n}  \theta_i \frac{1}{\rho_{i,p}} S_i
\end{equation}

The risk weights $(\theta_i)_{i=1..n}$ sums to one and are given by 
\begin{equation} 
\theta_i = \frac{w_i \rho_{i,p} \sigma_i}{\sigma_p}
\end{equation}

The coefficient $1 / \rho_{i,p}$ measures the diversification effect. It increases Sharpe ratio for low correlation. 
\end{proposition}

\begin{proof} given in \ref{proof:1}.
\end{proof} 

As a byproduct, we get the condition for a new asset to improve the overall portfolio Sharpe summarized below

\begin{proposition} \label{prop:2}
It is optimal to include an asset $i$ in a portfolio if and only if
\begin{equation} \label{eq:Sharpe}
S_i \geq \rho_{i,p} S_p
\end{equation}
\end{proposition} 

\begin{proof} See \ref{proof:2}.
\end{proof} 

This result is complementary from the standard mean variance approach as presented in \cite{Lobo_2007} and generalized in \cite{Shingo_2015}, which investigates about the optimal weights in a mean variance framework and states that the optimal weights are the result of a normal equation.

\subsection{General case}
A large number of performance ratios like Sharpe, Treynor, Sortino, Calmar, Sterling, information ratios or M2 write as the fraction of an excess return or a return or a return over a benchmark over some risk measure. The numerator and the denominator are homogeneous functions of degree 1. This leads to a performance ratio that is an homogeneous function of degree 0. In financial terms, the performance ratio is insensitive to leverage. 

We will write therefore any portfolio leverage-insensitive ratio $PR(p)$ as the fraction of an portfolio homogeneous return $R_p$ over an homogeneous function of degree 1, $f(p)$ as follows:
\begin{equation}
PR(p) = \frac{ R_p} {f(p)}
\end{equation}
We will also denote for the asset $i$ by $f(i)$ its denominator function 
and $PR(i)$ its coresponding performance ratio. 
Since the general return is an homogeneous functions of degree 1, 
it can be written as the convex combination of individual asset general returns:
\begin{equation}\label{eq:general1}
R_p = \sum_{i=1}^{n} w_i R_i
\end{equation}

Since the denominator is an homogeneous functions of degree 1, it can be written as a convex combination of individual asset contribution thanks to the Euler's homogeneous function theorem:

\begin{equation}\label{eq:general2}
f(p) = \sum_{i=1}^{n} w_i \frac{ \partial f}{\partial w_i}
\end{equation}

Combining equations \ref{eq:general1} and \ref{eq:general2} leads to a decomposition of the leverage-insensitive ratio $PR$ into individual or incremental performance ratio for each asset $i$

\begin{equation}\label{eq:general3}
PR(p) =  \sum_{i=1}^{n} \frac{  w_i    \frac{  \partial f}{\partial w_i} } { f(p) } \times \frac{ f(i)} {\frac{  \partial f}{\partial w_i}}  \times \frac{  R_i  } { f(i)} 
= \sum_{i=1}^{n} \theta_i \times D_i \times PR(i)
\end{equation}

The risk factor $\theta_i$ and the diversification factor $ D_i $ are given respectively by
\begin{equation}
\theta_i =  \frac{  w_i    \frac{  \partial f}{\partial w_i}}{ f(p)}, \quad  D_i =  \frac{ f(i)} {\frac{  \partial f}{\partial w_i}}  
\end{equation}

As in the case of the Sharpe, it is then easy to derive a condition for a new asset to improve the overall portfolio performance ratio summarized below
\begin{proposition} \label{prop:GeneralIncremental}
It is optimal to include an asset $i$ in a portfolio in order to maximize the performance ratio $PR(p)$ if and only if
\begin{equation} \label{eq:Sharpe}
PR(i)  \geq \frac{PR(p)}{D_i}   \Leftrightarrow  PR(i)  \geq \frac{\frac{  \partial f}{\partial w_i}}  { f(i)}   \times  PR(p)
\end{equation}
\end{proposition} 

\begin{proof} See \ref{proof:GeneralCase}.
\end{proof} 

The real work at this stage is to compute the derivative function of the denominator with respect to its weight $ \frac{ \partial f}{\partial w_i}$. We provide results for various performance ratios in the table below\\

\begin{table}[H]
\centering
\begin{tabular}{|l|c|c|c|}
  \hline
  performance  	& definition 					&   performance    								&  diversification   \\
   ratio			& 	= $\frac{  R_p}{f(p)}$	  	& marginal sensitivity = $\frac{  \partial f}{\partial w_i}$ 	& factor= $  f(i) / \frac{  \partial f}{\partial w_i} $ \\
  \hline 
  Sharpe 			& $S_p = \frac{r_p - r_f}{\sigma_p}$ 			& $\rho_{i,p} \sigma_i $  		&  $1 / \rho_{i,p} $ \\
  Sortino	 		& $Sor_p = \frac{r_p - r_f}{TSD_p}$ 			& $\rho_{i,p} TSD_i $			&  $ 1 / \rho_{i,p}$ \\
  Information 		& $IR_p = \frac{r_p - r_b}{\sigma_{p-b}}$ 	& $\rho_{i,p-b} \sigma_i $  		& $ 1 / \rho_{i,p-b}   \times \sigma_{i-b} / {  \sigma_i }$ \\
  Treynor 		& $T_p = \frac{r_p - r_f}{\beta_p}$ 			& $\beta_i $  					&  $1$ \\
  Recovery		& $Rec_p = \frac{r_p - r_f}{MDD_p}$ 		& $\widetilde{MDD}_i $  		&  $ MDD_i / \widetilde{MDD}_i $ \\
  Calmar	 		& $Cal_p = \frac{r_p - r_f}{MDD^{36m}_p}$   & $\widetilde{MDD}^{36m}_i $  	& $ MDD^{36m}_i / \widetilde{MDD}^{36m}_i $ \\
  Sterling 			& $Ster_p = \frac{r_p - r_f}{ALD_p}$ 			& $\widetilde{ALD}_i $  			& $ {ALD_i} /  \widetilde{ALD}_i  $  \\
  \hline
\end{tabular}
\caption{We provide here above the results for the most common performance ratios}
\label{Tabular1}
\end{table}

In table \ref{Tabular1}, we have used the following notations:
\begin{itemize}
\item TSD stands for target semi deviation(standard deviation of return below target). 
\item $r_b$ is the benchmark return.
\item $\sigma_{p-b}$ is the standard deviation of the difference between the portfolio and benchmark returns. 
\item  $MDD$ (respectively $MDD^{36m}$, $ALD$ ) stands for the maximum drawdown, the maximum drawdown over 36 months and the annual average maximum drawdown over the entire historical period.
\end{itemize}

\begin{proof} See \ref{proof:TabularProof}.
\end{proof}

\section{Numerical application}
\subsection{Sharpe ratio}\label{Sharpe_example}

Let us apply the above formulas to a portfolio consisting of three assets with the characteristics described in table \ref{Tab:port_characteristics}. The portfolio weights are the optimal ones in terms of the highest Sharpe ratio with the constraints of weights between 0 to 100\% (no short selling allowed neither extra leverage). We also provide in the characteristics the correlation between the asset $i$ and the portfolio as this is useful for risk decomposition.

\begin{table}[H]
\centering
\begin{tabular}{|l|c|c|c|c|}
  \hline
     Asset      				& I          	& II         	& III        	& Total \\
  \hline
    Weight     				& 34.87\%    	& 28.07\%    & 37.06\%   	& 100.00\% \\
    Expected Return 		& 3.20\%     	& 3.50\%     & 4.50\%     	&  \\
    Volatility 				& 4.87\%     	& 5.63\%     & 5.12\%     	&  \\
    Correlation with portfolio 	& 49.87\%   	& 37.55\%   & 65.87\%    	&  \\
  \hline
\end{tabular}
\caption{Portfolio characteristics}
\label{Tab:port_characteristics}
\end{table}

Once the characteristics established, we can easily compute the portfolio performance ratio as provided in table \ref{Tab:port_computation1}. We compute the portfolio return as the convex combination of the assets returns as well as the portfolio volatility. For the latter, we use the volatility reconstruction formula \ref{prop:port_vol}. The portfolio Sharpe is then the fraction of the latter two. We can notice that the resulting portfolio Sharpe (1.4000) is substantially higher than the best asset Sharpe (0.8789). We are benefiting fully from the diversification effect.

\begin{table}[H]
\centering
\begin{tabular}{|l|c|}
  \hline
               			& Portfolio \\
  \hline
  Expected Return 		& 3.77\% \\
  Volatility 			& 2.69\% \\
  Sharpe Ratio 		& 1.4000  \\
  \hline
\end{tabular}
\caption{Portfolio resulting Sharpe ratio: all these numbers are computed from table \ref{Tab:port_characteristics} }
\label{Tab:port_computation1}
\end{table}

The table \ref{Tab:port_computation2} gives us a nice view of the Portfolio Sharpe decomposition. The asset III has the highest Sharpe ratio (0.8789) but the lowest Sharpe diversification. 
This results in particular in a the highest risk (46.46\%), which is a strong indication that asset III contributes more to the overall portfolio Sharpe ratio.
 Its risk weight (46.46\%) is higher than its portfolio weight (37.06\%) indicating that the Sharpe contribution will be over-weighted. 
 In contrast, risk weights for asset I and II (31.48\% and 22.06\%) are smaller than their corresponding portfolio weights ( 34.87\% and 28.07\%).
  They will contribute less and will be under-weighted in the overall portfolio Sharpe ratio. 
  Thanks to the strong asset III contribution and the diversification, as noted above, the overall portfolio achieved a significant increase in its Sharpe ratio (1.4000). 
  In table \ref{Tab:port_computation2}, the Sharpe ratio relative contribution is defined as the Sharpe ratio contribution divided by the portfolio Sharpe ratio. It sums to 100 \%.

\begin{table}[H]
\centering
\begin{tabular}{|l|c|c|c|c|}
  \hline
    Asset & I     & II    & III   & Total \\
  \hline
    Asset Sharpe Ratio 		&  0.6571  	& 0.6217		& 0.8789		&  \\
    Sharpe Diversification 	& 2.0054  	& 2.6632  	& 1.5182  	&  \\
    Component Sharpe Ratio 	& 1.3177  	& 1.6557  	& 1.3344  	&  \\
    Risk Weight 			& 31.48\% 	& 22.06\% 	& 46.46\% 	& 100.00\% \\
    Sharpe Ratio Contribution 	& 0.4148 	& 0.3652 	& 0.6200 	& 1.4000  \\
    Sharpe Ratio Relative Contribution
    						& 29.63\% 	& 26.09\%	& 44.29\% 	& 100.00 \% \\
  \hline
\end{tabular}
\caption{Portfolio Sharpe decomposition: all these numbers are computed from table \ref{Tab:port_characteristics} }
\label{Tab:port_computation2}
\end{table}

\subsection{Recovery ratio}
Recovery ratio is an important performance ratio in the funds' world as it provides the expected return divided by the maximum drawdown. Maximum drawdown is closely monitored by professional investors as it gives an hint about the maximum potential loss should they invest and dis-invest at the worst time. For the sake of comparison with the previous study in section \ref{Sharpe_example}, we will first start with the same portfolio with the same percentage weights. For each asset, we provide in table \ref{Tab:recov_port_characteristics1} its maximum drawdown as well its performance marginal sensitivity (whose formula is $\frac{  \partial f}{\partial w_i}$) as provided in \ref{Tabular1}. 
\begin{table}[H]
\centering
\begin{tabular}{|l|c|c|c|c|}
  \hline
     Asset      				& I  			& II  		& III        	& Total \\
  \hline
    Weight     				& 34.87\%   	& 28.07\%  	& 37.06\%   	& 100.00\% \\
    Asset MDD 			& 5.71\% 	& 6.34\% 	& 4.53\% & \\
    Performance marginal sensitivity
    					 	& 3.14\% 	& 1.90\% 	& 4.07\%  & \\
  \hline
\end{tabular}
\caption{Portfolio characteristics for recovery ratio}
\label{Tab:recov_port_characteristics1}
\end{table}

Like for the Sharpe ratio, we can compute the portfolio resulting characteristics in table \ref{Tab:recov_port_computation1}. Expected return is like before computed as the convex combination of the asset returns (and is the same as in table \ref{Tab:port_computation1}). The recovery ratio is then simply the fraction of the latter two.

\begin{table}[H]
\centering
\begin{tabular}{|l|c|}
  \hline
               			& Portfolio \\
  \hline
  Expected Return 			& 3.77\% \\
    Portfolio Drawdown 		& 3.14\% \\
    Portfolio recovery ratio 	& 1.1999  \\
  \hline
\end{tabular}
\caption{Portfolio resulting recovery ratio}
\label{Tab:recov_port_computation1}
\end{table}

More interestingly is to analyze portfolio recovery decomposition as provided in table \ref{Tab:recov_port_computation1}. Again, thanks to portfolio diversification, we achieve a higher performance ratio (recovery of 1.1999) compared to the highest asset  performance ratio (obtained for asset III 0.9939). Like for the Sharpe ratio, the risk weight of asset III (48.12\%) is over-weighted compared to its portfolio weight (37.06\%). The opposite situation arises for asset II (risk weight of 17.02\% compared to a 28.07\%). By complete chance, risk and portfolio weight for asset I are equal up to the fourth decimal (risk weight of 34.865\% compared to portfolio weight of  34.872\% ). As in the case of the Sharpe ratio, we can check that the sum of the risk weights are equal to 100\%. 
In table \ref{Tab:recov_port_port_computation2}, the recovery ratio relative contribution is defined as the recovery ratio contribution divided by the portfolio recovery ratio. It sums to 100 \%.

\begin{table}[H]
\centering
\begin{tabular}{|l|c|c|c|c|}
  \hline
    Asset & I     & II    & III   & Total \\
  \hline
    Asset Recovery ratio 		&  0.5609  		& 0.5518  	& 0.9939  	&  \\
    Recovery Diversification 		&  1.8182  		& 3.3333  	& 1.1111  	&  \\
    Component Recovery 		&  1.0198  		& 1.8393  	& 1.1043  	&  \\
    Risk Weight 				&  34.87\% 		& 17.02\% 	& 48.12\% 	& 100.00\% \\
    Recovery Contribution 		&  0.3556		& 0.3130   	& 0.5314		& 1.1999  \\
    Recovery Relative Contribution &  29.63\%		& 26.08\%	& 44.28\%	& 100.00\% \\
  \hline
\end{tabular}
\caption{Portfolio recovery decomposition}
\label{Tab:recov_port_port_computation2}
\end{table}


A natural question that arises when looking at the recovery ratio for the portfolio is to determine if the optimal weights for the Sharpe ratio are also optimal for the recovery ratio. The answer is no in general. Recovery ratio is substantially different from Sharpe ratio. Hence the optimal portfolio for the recovery ratio has no reason to have the same weights as for the optimal portfolio for the Sharpe ratio. Because the recovery ratio implies a non convex function, namely the maximum drawdown, there is no closed form solution for the optimal portfolio as opposed to the Sharpe ratio settings. Using the CRG (that stands for Generalized Reduced Gradient) method  (as presented in \cite{Lasdon_1974}), we can determine the optimal weights for this portfolio as provided in table \ref{Tab:recov_port_characteristics2}. For this new portfolio, Asset performance marginal sensitivity changes slightly as the portfolio drawdown times are different.

\begin{table}[H]
\centering
\begin{tabular}{|l|c|c|c|c|}
  \hline
     Asset      				& I         	& II   		& III        	& Total \\
  \hline
   Weight 				& 5.99\% 	& 24.78\% 	& 69.24\% 	& \\
   Asset MDD 				& 5.71\% 	& 6.34\% 	& 4.53\% 	& \\
   Asset performance marginal sensitivity
						& 3.74\% 	& 1.95\% 	& 3.49\% 	& \\  
  \hline
\end{tabular}
\caption{Optimal Portfolio characteristics for maximum drawdown}
\label{Tab:recov_port_characteristics2}
\end{table}

We can recompute the new portfolio characteristics as provided in table \ref{Tab:recov_port_computation2}. We achieve a substantially higher portfolio recovery 
ratio (1.3379 versus  1.1999). This is due both to a higher expected return (4.17\%  versus 3.77\%) and a lower portfolio maximum drawdown  (3.12\% versus  3.14\%).

\begin{table}[H]
\centering
\begin{tabular}{|l|c|}
  \hline
							& Portfolio \\
  \hline
    Expected Return 			& 4.17\% \\
    Portfolio Drawdown 			& 3.12\% \\
    Portfolio recovery ratio 		& 1.3379 \\
  \hline
\end{tabular}
\caption{Optimal Portfolio resulting recovery ratio}
\label{Tab:recov_port_computation2}
\end{table}

As for previous studies, we can look at maximum drawdown decomposition as provided in table \ref{Tab:recov_port_port_computation2}. Compared to the previous portfolio with same weights as the optimal ones for the Sharpe ratio, the risk weight for asset III increases even more (77.38\%  versus  48.12\%). This is quite logical as this optimal portfolio for the maximum drawdown is indeed very much geared towards asset III (asset weight of 69.24\% versus 37.06\%). Interestingly, thanks to diversification, the recovery contribution for asset III (0.9986) is even higher to the asset recovery ratio (0.9939 ).

\begin{table}[H]
\centering
\begin{tabular}{|l|c|c|c|c|}
  \hline
    Asset 						& I     		& II    		& III   		& Total \\
  \hline
    Asset Recovery ratio 		& 0.5609  	& 0.5518  	& 0.9939  	&  \\
    Recovery Diversification 		& 1.5268  	& 3.2606  	& 1.2984  	&  \\
    Component Recovery 		& 0.8563  	& 1.7991  	& 1.2905  	&  \\
    Risk Weight 				& 7.17\% 	& 15.45\% 	& 77.38\% 	& 100.00\% \\
    Recovery Contribution 		& 0.0614		& 0.2779  	& 0.9986  	& 1.3379  \\
  \hline
\end{tabular}
\caption{Optimal Portfolio maximum drawdown decomposition}
\label{Tab:recov_port_port_computation2}
\end{table}

\section{Concluding Remarks}
We have introduced in this paper a unified framework for deriving asset contribution for performance ratios that are homogeneous function.
This allows us finding easily previous results on incremental Sharpe ratio contribution of a new asset as well as extend this to new performance 
ratios like Sortino, Information, Treynor, Recovery, Calmar or Sterling ratios where this did not exist. 
We also compare the impact of a new asset to a portfolio performance thanks to these incremental performance marginal sensitivity and show a methodology to analyse asset contribution to a portfolio. In a companion paper \cite{Benhamou_ConnectingSharpe}, we will analyze the statistical properties of the Sharpe
\clearpage

\appendix

\section{Various Proofs}
\subsubsection{Proof of Proposition \ref{prop:port_vol}}\label{proof:0}
Denoting by $\rho_{i,j}$ the correlation between asset $i$ and$j$, we can decompose the porfolio variance as a combination of assets' volatility as follows:
\begin{eqnarray} \label{portfolio_variance}
 \sigma_p^2 = \sum_{i=1...n} w_i^2 \sigma_i^2 + 2 \sum_{i,j=1...n, i \neq j} w_i w_j \rho_{i,j} \sigma_i \sigma_j
\end{eqnarray}

Differentiating the above equation \ref{portfolio_variance} with respect to $w_i$, we have
\begin{eqnarray} \label{diff_portfolio_variance}
2  \sigma_p \frac{ \partial \sigma_p} { \partial w_i}   =  2  w_i \sigma_i^2 + 2 \sum_{j=1...n, j \neq i} w_j \rho_{i,j} \sigma_i \sigma_j 
\end{eqnarray}

We can notice that the correlation between asset $i$ and the portfolio $p$  is given by
\begin{eqnarray} \label{diff_portfolio_variance}
\rho_{i,p} = \frac{     w_i \sigma_i^2 + \sum_{j=1...n, j \neq i} w_j \rho_{i,j} \sigma_i \sigma_j  } { \sigma_i \sigma_p} 
\end{eqnarray}

which shows that 
\begin{eqnarray} \label{diff_portfolio_variance}
\frac{ \partial \sigma_p} { \partial w_i}  = \rho_{i,p} \sigma_i
\end{eqnarray}

Since the portfolio volatility is homogeneous of degree 1, the Euler's homogeneous function theorem states that 

\begin{eqnarray}  
\sigma_p   =  \sum_{i=1...n} w_i \frac{ \partial \sigma_p} { w_i} = \sum_{i=1...n} w_i  \rho_{i,p} \sigma_i  
\end{eqnarray}
\qed

\subsubsection{Proof of Proposition \ref{prop:1}}\label{proof:1}
Dividing and multiplying by $\rho_{i,p} \sigma_i$ in the formula of the portfolio Sharpe ratio and regrouping the terms leads to the final results as follows:

\begin{eqnarray}
S_p =  \sum_{i=1}^{n}  w_i   \frac{r_i}{\sigma_p} =  \sum_{i=1}^{n}  \frac{w_i \rho_{i,p} \sigma_i }{\sigma_p}  \frac{1}{\rho_{i,p}}   \frac{r_i}{\sigma_i} =  \sum_{i=1}^{n}  \theta_i  \frac{1}{\rho_{i,p}} S_i
\end{eqnarray}
\qed

\subsubsection{Proof of Proposition \ref{prop:2}}\label{proof:2}

Let us denote by $P(w_i, i=1..n)$ the portfolio composed of $n$ assets with percentage weights $w_i$ and $n$ the new asset. The portfolio percentage weights sum to 1: $\sum_{i=1..n} w_i = 1$. The optimization program writes as follows:
\begin{eqnarray}
\text{maximize} & &  \text{Sharpe Ratio}( P(w_i, i=1..n) ),   \\
\text{subject to} & & \sum_{i=1..n} w_i = 1
\end{eqnarray}

Using proposition \ref{prop:1} and multiplying and dividing by $1- \theta_n$ (with the additional constraint that $\theta_n \neq 1$ \footnote{the particular case of $\theta_n =1$ can be handled easily by taking the left limit for $ \theta_n  \uparrow 1$}), the optimal solution is also the solution of this program

\begin{eqnarray}
\text{maximize} & & (1- \theta_n) \sum_{i=1..n-1}  \frac{\theta_i}{1- \theta_n}  \frac{1}{\rho_{i,P}} S_i +  \theta_n  \frac{1}{\rho_{n,P}} S_n \\
\text{subject to} & & \sum_{i=1..n} \theta_i = 1, \quad \theta_n \neq 1
\end{eqnarray}

Fixing $\theta_n$ and noticing that the weights $ \frac{\theta_i}{1- \theta_n}$ for $i=1..n-1$ sum to 1, the optimization program is indeed a two steps program where we can optimize first in terms of the weights $ \frac{\theta_i}{1- \theta_n}$ and then in terms of $ \theta_n$. As the $n-1$ terms are indeed the percentage weights of an $n-1$ portfolio composed of $n-1$ assets, the first optimization is exactly the same as the optimization of the optimal portfolio with $n-1$ assets in terms of its Sharpe ratio. The first step therefore leads to the optimal portfolio without asset $n$ for the Sharpe ratio. We will denote this portfolio by $\widetilde{P}$. 
The maximization program is then equivalent to 
\begin{eqnarray}
\text{maximize} & & (1- \theta_n) S_{\widetilde{P}} +  \theta_n  \frac{1}{\rho_{n,P}} S_n \\
\text{subject to} & & 0 \leq \theta_n \leq  1
\end{eqnarray}

This optimization program is a linear function whose optimal solution $\theta_n$ is not equal to zero if and only if the slope coefficient is positive (which proves the result):

\begin{eqnarray}
 \frac{1}{\rho_{n,P}} S_n  -  S_{\widetilde{P}} \geq  0  \Leftrightarrow  S_n \geq   \rho_{n,P} S_{\widetilde{P}}  
\end{eqnarray}

\qed

\subsubsection{Proof of Proposition \ref{prop:GeneralIncremental}}\label{proof:GeneralCase}
The proof is exactly the same as in \ref{proof:2} and leads at the end to solve the following linear maximization program where like in \ref{proof:2} , we denote by $\widetilde{P}$ the optimal portfolio with $n-1$ assets in terms of the performance ratio.

\begin{eqnarray}
\text{maximize} & & (1- \theta_n) PR(\widetilde{P}) +  \theta_n  D_n PR(n) \\
\text{subject to} & & 0 \leq \theta_n \leq  1
\end{eqnarray}

This optimization program is a linear function whose optimal solution $\theta_n$ is not equal to zero if and only if the slope coefficient is positive (which proves the result):

\begin{eqnarray}
 D_n PR(n)  - PR(\widetilde{P})  \geq  0  \Leftrightarrow & PR(n)  \geq  \frac{PR(\widetilde{P})}{D_n}
\end{eqnarray}

\qed

\subsubsection{Proof of Table \ref{Tabular1} results}\label{proof:TabularProof}
Results for Sharpe are already proved in \ref{proof:0}. \par
Denoting by $r_M$ the return of the market asset, the beta in the Treynor ratio is given by
\begin{eqnarray}
\beta_p = \frac{Cov(r_p,r_M)}{\sigma_M^2} = \frac{Cov( \sum_{i=1}^n w_i r_i, r_M)}{\sigma_M^2} = \sum_{i=1}^n  \frac{Cov(r_i, r_M)}{\sigma_M^2}
\end{eqnarray}

A straight derivation leads to the results $ \frac{  \partial f}{\partial w_i} = \beta_i $\par
The target semi deviation is quite similar to the standard deviation with the additional constraint that we only use returns that are below its mean. The proof is therefore similar to the one of the portfolio in the Sharpe ratio with the additional constraint to use only down returns, leading to a target standard deviation for the assets' performance marginal sensitivity: 
\begin{eqnarray}
\frac{  \partial f}{\partial w_i} = \rho_{i,p} TSD_i
\end{eqnarray}\par

The proof for the recovery, Calmar and Sterling ratio are similar and we will detail only the first one as the other derivation are just an averaging or windowing of the first proof.
The maximum drawdown measures the largest peak-to-trough decline in the value of a portfolio. Denoting by $0$ to $T$ the historical times at which we observe the portfolio return and by $cr_p^j$ the cumulative return of the portfolio from time $0$ to $j$ with the convention that the return at time $0$ is null, the maximum drawdown can be written mathematically as the maximum of the following discrete optimization program

\begin{eqnarray}
\text{MDD} \equiv \min_{j=0..T,  \>\> k =j..T} \frac {1+ cr_p^k } {1+ cr_p^j }-1
\end{eqnarray}\par

The portfolio contains $n$ assets with percentage weights $w_i$. Assuming no re-balancing, the portfolio cumulative returns writes as the convex combination of the cumulative asset $i$ return, leading to the following definition of the maximum drawdown 
\begin{eqnarray}
\text{MDD} = \min_{j=0..T,  \>\> k =j..T} \frac {\sum_{i=1}^n w_i (1+  cr_i^k) } {\sum_{i=1}^n w_i (1+  cr_i^j) }-1
\end{eqnarray}\par

As this is a discrete optimization program, the optimum is attained for $j_*$ and $k_*$. Deriving the maximum drawdown leads therefore to

\begin{eqnarray}
\frac{\partial} {\partial w_i} \text{MDD} = \left(  \frac{1+  cr_i^{k_*}} {1+ cr_p^{j_*}}-1 \right)  -\left(   \frac{ \frac{1+ cr_i^{j_*}}   {1+ cr_p^{j_*}} \times (1+ cr_p^{k_*} )}  {1+ cr_p^{j_*}} -1\right)  \equiv \widetilde{MDD}_i
\end{eqnarray}\par

Hence the sensitivity of the maximum drawdown is given by the difference of 
\begin{itemize}
\item the drawdown between between the portfolio cumulative return at time $j_*$ and the asset at time $k_*$
\item the drawdown between between the portfolio cumulative return at time $j_*$ and the portfolio return at time $k_*$ augmented by the difference of cumulative return between the asset $i$ and portfolio at time $j_*$
\end{itemize}

For the  Calmar and Sterling ratio, similar formulas exist where the sensitivity of the maximum drawdown is taken over 36 months, respectively as the annual average.
\qed

\subsubsection{Correlation matrix for the three assets}\label{Correl_mat}
For the sake of completeness, we provide below the correlation matrix. This matrix is consistent with asset correlation with portfolio coefficients.
\begin{table}[H]
\centering
    \begin{tabular}{|l|c c c |}
    \hline
     Asset    &  I        &  II       &  III  \\
    \hline
     I        &                                              1.00  & -                 0.20  &                   0.40  \\
     II       & -                                            0.20  &                   1.00  &                   0.30  \\
     III      &                                              0.40  &                   0.30  &                   1.00  \\
    \hline
    \end{tabular}
  \caption{Asset correlation matrix}
  \label{tab:correl_mat}
\end{table}

From this correlation matrix denoted by $\Sigma$ and for asset $i$ with the corresponding Kronecker delta vector defined by $\delta_i = (0 ... 1 ... 0 )^T$ with one at the ith row and zero elsewhere, it is then straightforward to compute for any asset its correlation with portfolio (whose weight vector is defined as $W = (w_1, .. w_j, .., w_n)^T$ as follows:

\begin{eqnarray}
\rho_{i,p} = \frac{ \delta_{i} \Sigma W }{ \sqrt{\delta_{i} \Sigma \delta_{i}}  \sqrt{W  \Sigma W} }
\end{eqnarray}
\clearpage


\bibliographystyle{jfe}
\bibliography{mybib}

\end{document}